\documentclass[11pt]{article}

\usepackage[letterpaper,margin=1in]{geometry}
\usepackage{amsmath,amssymb,mathtools}
\usepackage{amsthm}
\usepackage{algorithm}
\usepackage[noend]{algpseudocode}
\usepackage{enumitem}
\usepackage{microtype}
\usepackage{xcolor}
\usepackage{tikz}
\usepackage{float}
\usetikzlibrary{arrows.meta,positioning}
\usepackage[colorlinks=true,linkcolor=blue,citecolor=blue,urlcolor=blue]{hyperref}

\makeatletter
\newenvironment{breakablealgorithm}
  {%
   \begin{center}
     \refstepcounter{algorithm}%
     \hrule height.8pt depth0pt \kern2pt
     \renewcommand{\caption}[2][\relax]{%
       {\raggedright\textbf{\ALG@name~\thealgorithm} ##2\par}%
       \ifx\relax##1\relax
         \addcontentsline{loa}{algorithm}{\protect\numberline{\thealgorithm}##2}%
       \else
         \addcontentsline{loa}{algorithm}{\protect\numberline{\thealgorithm}##1}%
       \fi
       \kern2pt\hrule\kern2pt
     }%
  }
  {%
     \kern2pt\hrule\relax
   \end{center}
  }
\makeatother

\setlist[itemize]{leftmargin=2.25em,itemsep=0.15em,topsep=0.35em}

\theoremstyle{plain}
\newtheorem{theorem}{Theorem}[section]
\newtheorem{lemma}[theorem]{Lemma}

\newtheorem{corollary}[theorem]{Corollary}

\newtheorem{observation}[theorem]{Observation}

\theoremstyle{definition}
\newtheorem{definition}[theorem]{Definition}

\theoremstyle{remark}

\newtheorem*{remark*}{Remark}

\newcommand{\Active}{I}
\newcommand{\Bid}{\operatorname{bid}}
\newcommand{\Cost}{\operatorname{cost}}

\title{Full Justified Representation under Hare and Droop Quotas in Polynomial Time}

\author{Yizhou Ai\\University of Toronto\\ \texttt{yizhou.ai@mail.utoronto.ca}}
\date{}
\begin{document}
\maketitle

\begin{abstract}
    I study \emph{Full Justified Representation} (FJR) in approval-based multiwinner elections under both the Hare and Droop quota conventions. I introduce a descending-budget algorithm in which voters distribute their remaining budgets across their current representation gaps and candidates are purchased whenever the resulting offers cover a common price. With candidate price $\lambda_H=n/k$, the algorithm returns a Hare-FJR committee; with candidate price $\lambda_D=n/(k+1)$, it returns a committee satisfying the more demanding Droop-FJR axiom of Casey and Elkind. The two guarantees share a historical-payment invariant and a terminal row--column accounting argument, while the Droop proof requires a new residual-budget argument when all $k$ paid seats are filled. Both variants are deterministic once the voter and candidate orders are fixed and use $O(kmn)$ rational operations.
\end{abstract}

\section{Introduction}\label{sec:introduction}

Approval-based multiwinner elections ask a set of voters to select a fixed-size committee from a set of candidates, where each voter reports the candidates she approves. The model has been studied extensively as a framework for representative collective choice; Lackner and Skowron provide a comprehensive treatment of its principal rules, axioms, algorithms, and applications \cite{lackner2023multiwinner}. A central objective is proportional representation: sufficiently large groups of voters with sufficiently aligned preferences should receive a commensurate amount of representation in the selected committee. The justified-representation framework formalizes this objective without assuming that the electorate is partitioned into parties. Its basic axiom, \emph{Justified Representation} (JR), protects a group of at least one Hare quota of voters when the group agrees on a candidate; stronger axioms such as \emph{Proportional Justified Representation} (PJR) and \emph{Extended Justified Representation} (EJR) scale this guarantee with the size and cohesiveness of the group \cite{aziz2017jr,fernandez2017pjr}. Subsequent work has studied both the complexity of constructing and verifying such committees and their empirical prevalence across approval profiles \cite{aziz2018ejr,bredereck2019experimental}.

A useful way to organize these axioms is to separate two modeling choices: what makes a coalition sufficiently \emph{cohesive}, and what it means for that coalition to be \emph{represented}. PJR and EJR use the standard notion of $\ell$-cohesiveness, which requires a coalition of at least $\ell$ Hare quotas to unanimously approve at least $\ell$ common candidates. \emph{Full Justified Representation} (FJR) weakens this common-intersection requirement. A coalition may instead identify an alternative set $T$ that is affordable with its proportional share and from which every coalition member approves at least $\ell$ candidates, although the approved candidates may differ across voters. FJR then applies the individual-witness representation requirement of EJR: at least one voter in the coalition must obtain utility at least $\ell$ from the selected committee. In this sense, FJR protects groups whose agreement is distributed across a small, densely approved witness set rather than concentrated on the same candidates \cite{peters2021proportional,KalayciLK2025}.

FJR was introduced by Peters, Pierczy\'nski, and Skowron in the more general participatory-budgeting setting. They showed that the \emph{Method of Equal Shares} (MES), despite satisfying EJR, may violate FJR, and established the existence of FJR outcomes through the computationally demanding \emph{Greedy Cohesive Rule} (GCR) \cite{peters2021proportional}. Aziz et al. later incorporated GCR into a best-of-both-worlds construction that outputs a lottery satisfying ex-post FJR and ex-ante Strong UFS, but their FJR construction is not polynomial-time computable \cite{aziz2023bestofbothworlds}. Kalayci, Liu, and Kempe introduced \emph{Full Proportional Justified Representation} (FPJR), which combines the weak cohesiveness notion of FJR with the collective-representation requirement of PJR. They showed that several efficient rules satisfy FPJR, while PAV may violate it, and emphasized that no polynomial-time rule was known to guarantee FJR \cite{KalayciLK2025}. Casey and Elkind subsequently developed Droop-quota counterparts of the JR axioms and showed that a modification of GCR satisfies Droop-FJR; this result likewise relies on GCR, which is not known to be polynomial-time computable \cite{casey2026droop}. Thus, efficient construction remained open for both Hare-FJR and the more demanding Droop-FJR axiom.

The difficulty is structural rather than merely implementational. An FJR violation is witnessed by a triple $(S,T,\ell)$, and the definition quantifies simultaneously over voter coalitions, alternative candidate sets, and representation thresholds. The witness set $T$ need not contain $\ell$ candidates approved by every voter in $S$, so the common-candidate arguments used for standard cohesive groups do not apply directly. Moreover, verifying whether a given committee satisfies FJR is coNP-complete \cite{KalayciLK2025}; consequently, a construction cannot be expected to rely on repeatedly recognizing and eliminating arbitrary FJR violations by a direct polynomial-time subroutine. A successful proof must instead rule out all witnesses implicitly and must account jointly for payments already made to candidates in $T\cap W$ and the coalition's remaining ability to support candidates in $T\setminus W$.

\subsection{Results and Technical Overview}\label{subsec:results-overview}

This work gives deterministic polynomial-time constructions for both Hare-FJR and Droop-FJR in approval-based multiwinner elections. The same descending-budget algorithm is used in both cases; only the common candidate price changes. For the Hare version I set
\[
    \lambda_H=\frac{n}{k},
\]
whereas for the Droop version I set
\[
    \lambda_D=\frac{n}{k+1}.
\]
Given fixed voter and candidate orders, each instantiation deterministically returns a committee of size exactly $k$. A direct implementation performs $O(kmn)$ rational operations.


The construction combines two ideas from existing polynomial-time rules. From the \emph{Greedy Justified Candidate Rule} (GJCR), I take the descending-level organization, processing representation targets in the order
\[
    h=k,k-1,\ldots,1.
\]
From the \emph{Method of Equal Shares} (MES), I take the virtual-budget interpretation: every voter begins with one unit of budget, and a candidate is selected only when the offers of its active approvers cover the common price $\lambda$~\cite{brill2023robust,peters2020proportionality,peters2021proportional}. Neither ingredient alone suffices for FJR. GJCR reasons about the underrepresented supporters of a single candidate, while ordinary MES-style payments do not record how many additional approved candidates a voter needs to reach a common representation target.

The new ingredient is \emph{gap-normalized purchasing power}. At a decision snapshot in phase $h$, a voter $i$ is active if she approves at least $h$ candidates in total but currently receives fewer than $h$ approved candidates from the paid set $P$. If her remaining balance is $r_i$, she offers
\[
    \frac{r_i}{h-u_i(P)}
\]
to every unselected candidate she approves. Thus, a voter with a larger representation gap spreads her remaining budget more conservatively, while a voter closer to the target can place a larger offer. Whenever the aggregate offers for a candidate reach $\lambda$, the candidate is purchased for exactly $\lambda$, with each contributing voter charged no more than her current offer. The state is then recomputed. If no candidate is affordable at level $h$, the algorithm lowers the target to $h-1$. Any unfilled seats are added by deterministic padding after the paid construction stops.

The common part of the proof rules out FJR witnesses implicitly. Its central invariant states that, whenever a voter is active at a later decision snapshot, her current offer is at least every payment she made for an earlier selected candidate. This permits a mixed accounting scheme in which selected candidates are charged at their historical payments and unselected candidates at the current offer. Every active voter assigns mixed cost at least one to any approved set containing at least $h$ candidates.

At a normal terminal snapshot of phase $\ell$, summing these mixed costs by voter gives a lower bound of $|S|$ for a putative witness $(S,T,\ell)$. Summing the same quantities by candidate gives a strict upper bound of $\lambda|T|$. For Hare-FJR this contradicts $|S|\ge\lambda_H|T|$, while for Droop-FJR it contradicts the strict condition $|S|>\lambda_D|T|$.

The only substantive difference arises when the paid set reaches size $k$ before phase $\ell$ has a normal terminal snapshot. Under the Hare price, $k\lambda_H=n$, so all voter budgets are exhausted and every voter reaches the target of the final purchase phase. Under the Droop price, the total residual budget is instead
\[
    n-k\lambda_D=\frac{n}{k+1}.
\]
I show that any alleged Droop-FJR witness with threshold at most the final purchase phase would force the witness voters alone to retain more than this amount, yielding a contradiction. This replacement lemma completes the Droop-FJR proof without changing the algorithm beyond its candidate price.

The formal construction is presented in Section~\ref{sec:algorithm}, its implementation and computational complexity are analyzed in Section~\ref{sec:complexity}, and the two FJR guarantees are proved in Section~\ref{sec:proof}.
\subsection{Related work and the JR hierarchy}\label{subsec:related-work}

\paragraph{The classical Hare-quota hierarchy and terminology.}
Under the standard Hare-quota definitions, JR, PJR, and EJR differ in the representation demanded from a standard $\ell$-cohesive group. JR is the one-seat base case; PJR requires that the group collectively approve at least $\ell$ elected candidates; and EJR requires an individual witness in the group who approves at least $\ell$ elected candidates \cite{aziz2017jr,fernandez2017pjr}.

Following the terminology of Ai and Tao, I use \emph{Strong Justified Representation} (SJR) for the requirement that every member of every $\ell$-cohesive group approves at least $\ell$ elected candidates. The same all-level condition was introduced by Brill et al.\ under the name \emph{Individual Representation} (IR). The earlier notion of \emph{semi-strong justified representation}, introduced by Aziz et al., imposes the corresponding individual guarantee only for $1$-cohesive groups. Ai and Tao explicitly adopt the name SJR for the full all-$\ell$ condition \cite{aziz2017jr,brill2025individual,ai2026computational}.

\emph{Average Justified Representation} (AJR) requires the average utility of every $\ell$-cohesive group to be at least $\ell$. As Ai and Tao note, S\'anchez-Fern\'andez et al.\ studied this average-satisfaction condition but did not use the name AJR; Ai and Tao follow the terminology used by Han et al.\ \cite{fernandez2017pjr,han2026likelihood,ai2026computational}. Under these conventions,
\[
    \mathrm{SJR}\Longrightarrow\mathrm{AJR}\Longrightarrow
    \mathrm{EJR}\Longrightarrow\mathrm{PJR}\Longrightarrow\mathrm{JR}.
\]
Unlike the axioms from EJR downward, SJR and AJR committees need not exist. Ai and Tao show that deciding whether an SJR committee exists is $\Theta_2^p$-complete, whereas deciding whether an AJR committee exists is $\Sigma_2^p$-complete \cite{ai2026computational}. The computational study of the classical hierarchy shows a notable separation between construction and verification: EJR committees can be found in polynomial time, whereas verifying EJR or PJR is coNP-complete \cite{aziz2018ejr}. Complementary work quantifies proportionality through proportionality-degree guarantees and studies the complexity of computing or verifying these guarantees \cite{skowron2021proportionality,janeczko2022proportionality}, while experimental studies examine the prevalence and diversity of JR-, PJR-, and EJR-satisfying committees \cite{bredereck2019experimental}.

\paragraph{Weak cohesiveness and full representation.}
The original Hare-quota version of FJR changes the cohesiveness side of the definition. Instead of requiring $\ell$ candidates approved unanimously by the coalition, it allows a witness set $T$ that is affordable with the coalition's proportional share and from which every coalition member approves at least $\ell$ candidates. FJR combines this weaker cohesiveness condition with EJR's individual-witness representation requirement \cite{peters2021proportional}. FPJR uses the same weak cohesiveness condition but replaces the individual-witness conclusion by PJR's collective-representation conclusion \cite{KalayciLK2025}. The best-of-both-worlds literature also considers FJR as an ex-post guarantee for randomized committee selection: Aziz et al. obtain a lottery satisfying ex-post FJR together with ex-ante Strong UFS, although their FJR construction is not polynomial-time computable \cite{aziz2023bestofbothworlds}. Casey and Elkind systematically study versions of JR, PJR, EJR, FPJR, FJR, PJR+, and EJR+ based on the more demanding Droop quota. They prove that a modified GCR satisfies Droop-FJR, but this rule is not known to be polynomial-time computable \cite{casey2026droop}. The Droop-price instantiation developed here gives a polynomial-time construction for their Droop-FJR axiom.

\paragraph{Core stability.}
Within the classical Hare-quota hierarchy, \emph{Core Stability} is stronger than FJR because it compares every coalition with every proportionally affordable alternative and requires some coalition member to weakly prefer the selected committee. Thus,
\[
    \text{Core Stability}\Longrightarrow\mathrm{FJR}\Longrightarrow\mathrm{EJR},
    \qquad
    \mathrm{FJR}\Longrightarrow\mathrm{FPJR}\Longrightarrow\mathrm{PJR}.
\]
The core has generated a parallel literature on group fairness and stability. Stable lotteries are known to exist under approval utilities \cite{cheng2020group}, while deterministic work studies approximately stable committees \cite{jiang2020approximately,munagala2022approximate}, exact core existence on restricted preference domains \cite{pierczynski2022core}, and existence for bounded numbers of seats, candidates, or voter types \cite{peters2025core,berker2026edge,becker2026core}. General core non-emptiness for approval-based multiwinner elections remains an important open problem.

\paragraph{The ``plus'' branch.}
EJR+ and PJR+ relax cohesiveness in a different direction: they consider a sufficiently large group that unanimously approves at least one unelected candidate, rather than a group possessing an $\ell$-candidate common intersection or an FJR witness set. EJR+ applies the individual-witness conclusion, while PJR+ applies the collective conclusion. Brill and Peters show that
\[
    \mathrm{EJR+}\Longrightarrow\mathrm{EJR},
    \qquad
    \mathrm{EJR+}\Longrightarrow\mathrm{PJR+}\Longrightarrow\mathrm{PJR},
\]
and give the polynomial-time GJCR algorithm for EJR+ \cite{brill2023robust}. The branches are not totally ordered: in particular, EJR+ and FJR are incomparable, and FPJR is incomparable with EJR as well as with the two plus axioms \cite{brill2023robust,KalayciLK2025}.

Figure~\ref{fig:jr-hierarchy} summarizes the implication relations among the Hare-quota axioms discussed above. The diagram is adapted from the hierarchy in Ai and Tao~\cite{ai2026computational}. Droop-quota variants are not included. An arrow $X\to Y$ means that every committee satisfying $X$ also satisfies $Y$; omitted arrows are not intended to assert an implication.

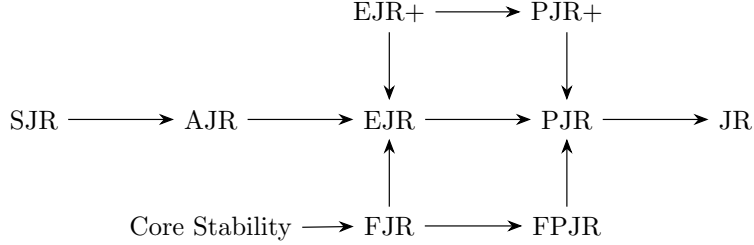
\begin{figure}[H]
    \centering
    \begin{tikzpicture}[
        node distance=9mm and 14mm,
        every node/.style={font=\small},
        implication/.style={-{Stealth[length=2mm]},line width=0.45pt}
        ]
        \node (sjr) {SJR};
        \node[right=of sjr] (ajr) {AJR};
        \node[right=of ajr] (ejr) {EJR};
        \node[right=of ejr] (pjr) {PJR};
        \node[right=of pjr] (jr) {JR};

        \node[above=of ejr] (ejrp) {EJR+};
        \node[above=of pjr] (pjrp) {PJR+};

        \node[below=of ajr] (core) {Core Stability};
        \node[below=of ejr] (fjr) {FJR};
        \node[below=of pjr] (fpjr) {FPJR};

        \draw[implication] (sjr) -- (ajr);
        \draw[implication] (ajr) -- (ejr);
        \draw[implication] (ejr) -- (pjr);
        \draw[implication] (pjr) -- (jr);

        \draw[implication] (ejrp) -- (ejr);
        \draw[implication] (ejrp) -- (pjrp);
        \draw[implication] (pjrp) -- (pjr);

        \draw[implication] (core) -- (fjr);
        \draw[implication] (fjr) -- (ejr);
        \draw[implication] (fjr) -- (fpjr);
        \draw[implication] (fpjr) -- (pjr);
    \end{tikzpicture}
    \caption{Implication relations among the Hare-quota justified-representation axioms discussed in this paper. Droop-quota variants are not shown.}
    \label{fig:jr-hierarchy}
\end{figure}

\paragraph{Algorithms and market-based approaches.}
\emph{Proportional Approval Voting} (PAV) satisfies EJR but is computationally difficult to optimize exactly \cite{aziz2017jr}, while \emph{Sequential Phragm\'en} satisfies PJR and may fail EJR~\cite{brill2024phragmen}. The standard argument establishing EJR for PAV relies crucially on the unanimity built into ordinary $\ell$-cohesiveness. If every voter in an $\ell$-cohesive group $S$ has utility below $\ell$, then some unelected candidate is approved by every voter in $S$; this candidate provides the per-candidate harmonic marginal-gain bound used in the PAV exchange argument. Weak FJR cohesiveness does not provide such a candidate. It only guarantees an affordable witness set $T$ from which every voter approves at least $\ell$ candidates, and these approved subsets may differ across voters. After accounting for candidates already selected, the coalition's remaining approved alternatives may therefore be dispersed across $T\setminus W$, so no single candidate need satisfy the marginal-gain bound required by the EJR argument. Indeed, PAV may violate FPJR~\cite{KalayciLK2025} and hence does not guarantee FJR, since FJR implies FPJR. Thus, the known welfare and local-exchange proofs for EJR do not directly extend to FJR; excluding FJR violations requires an argument that accounts for the witness set $T$ collectively.

Market-based approaches provide a different account of proportionality. MES and related priceability concepts interpret voters as holding budgets and candidates as having prices \cite{peters2020proportionality}; stable priceability and Lindahl-style explanations further connect proportional committees with public-goods markets \cite{PetersP0021}. Priceable committees can also be completed while retaining strong coverage and utilitarian guarantees \cite{brill2024completing}. Related work extends price-based and verifiable proportionality ideas to approval-based participatory budgeting \cite{brill2023proportionality,kraiczy2023adaptive}. In the present construction, ordinary fundability is augmented with threshold information: the gap-normalized offer $r_i/(h-u_i(P))$ records how many further approved candidates a voter needs to reach the current target. This combines the market interpretation of MES with the descending-level organization of GJCR and closes the efficient-construction gap left by GCR and the coNP-completeness of FJR verification \cite{peters2021proportional,KalayciLK2025}.

\subsection{Concurrent Work}\label{subsec:concurrent-work}

I developed the algorithm and proofs in this paper independently of the concurrent work of Frank and Peters~\cite{frank2026polynomial}.
Their work also gives a polynomial-time rule satisfying Hare-FJR.

\section{Preliminaries}\label{sec:preliminaries}

In this section, I introduce approval-based multiwinner elections and the Hare- and Droop-FJR demands used in the two guarantees. The other justified-representation axioms are discussed only in Section~\ref{subsec:related-work}, since they are not used in the algorithm or its proof.

\paragraph{Terminology.}
Unless a Droop variant is explicitly specified, the justified-representation axioms discussed in this paper use their standard Hare-quota definitions. In particular, the unqualified term \emph{FJR} refers to the original axiom of Peters, Pierczy\'nski, and Skowron. When comparing the two quota conventions, I call it \emph{Hare-FJR}; the variant of Casey and Elkind is always called \emph{Droop-FJR}.

\begin{definition}[Approval-based multiwinner election]\label{def:election}
    An \emph{approval-based multiwinner election} is a tuple
    \[
        E=(N,C,\mathcal A,k),
    \]
    where:
    \begin{itemize}
        \item $N=[n]=\{1,2,\ldots,n\}$ is a finite, nonempty set of voters;
        \item $C$ is a finite set of candidates, with $m=|C|$;
        \item $\mathcal A=(A_1,\ldots,A_n)$ is the approval profile, where $A_i\subseteq C$ is the set of candidates approved by voter $i$; and
        \item $k\in\mathbb Z_{>0}$ is the target committee size, with $1\le k\le m$.
    \end{itemize}
\end{definition}

\begin{definition}[Committee]\label{def:committee}
    Given an election $E=(N,C,\mathcal A,k)$, a \emph{committee} is a subset $W\subseteq C$ of size $|W|=k$.
\end{definition}

\begin{definition}[Approval utility]\label{def:utility}
    For every voter $i\in N$ and candidate set $X\subseteq C$, the \emph{approval utility} (or satisfaction) of voter $i$ from $X$ is
    \[
        u_i(X):=|A_i\cap X|.
    \]
\end{definition}

Define the normalized Hare and Droop thresholds by
\[
    \lambda_H:=\frac{n}{k},
    \qquad
    \lambda_D:=\frac{n}{k+1}.
\]
The conventional single-seat integer Droop quota is $\lfloor n/(k+1)\rfloor+1$. Casey and Elkind define Droop $\ell$-cohesiveness and Droop weak cohesiveness using the strict thresholds $\ell n/(k+1)$ and $|T|n/(k+1)$, respectively, rather than by multiplying the rounded single-seat quota~\cite{casey2026droop}. I follow their formal definition; the corresponding normalized quantity $\lambda_D=n/(k+1)$ is also the candidate price in the Droop instantiation below.

\begin{definition}[Hare and Droop weak cohesiveness]\label{def:weak-cohesiveness}
    Let $1\le \ell\le k$, let $S\subseteq N$ be nonempty, and let $T\subseteq C$. The group $S$ is
    \begin{itemize}
        \item \emph{Hare weakly $(\ell,T)$-cohesive} if
              \[
                  |S|\ge \lambda_H|T|
                  \qquad\text{and}\qquad
                  u_i(T)\ge\ell\quad\text{for every }i\in S;
              \]
        \item \emph{Droop weakly $(\ell,T)$-cohesive} if
              \[
                  |S|> \lambda_D|T|
                  \qquad\text{and}\qquad
                  u_i(T)\ge\ell\quad\text{for every }i\in S.
              \]
    \end{itemize}
    In either case, $T$ is called a \emph{witness set}. A triple $(S,T,\ell)$ satisfying the corresponding conditions is called a Hare-FJR demand or a Droop-FJR demand, respectively.
\end{definition}

The utility condition does not require the voters to approve the same $\ell$ candidates: each voter may obtain her $\ell$ approved candidates from a different part of the common witness set $T$. This distributed agreement is the weak-cohesiveness feature that distinguishes FJR from EJR.

\begin{definition}[Hare-FJR and Droop-FJR]\label{def:fjr}
    A committee $W$ satisfies \emph{Hare Full Justified Representation} (Hare-FJR) if, for every Hare-FJR demand $(S,T,\ell)$, there exists a voter $i\in S$ such that
    \[
        u_i(W)\ge\ell.
    \]
    It satisfies \emph{Droop Full Justified Representation} (Droop-FJR) if the same conclusion holds for every Droop-FJR demand. The latter definition is due to Casey and Elkind~\cite{casey2026droop}.
\end{definition}

\begin{observation}
    Every Droop-FJR committee satisfies Hare-FJR.
\end{observation}

\begin{proof}
    Every FJR demand has a nonempty witness set. Since
    \[
        \frac{n}{k}|T|>\frac{n}{k+1}|T|,
    \]
    every Hare-FJR demand is also a Droop-FJR demand. Hence, Droop-FJR implies Hare-FJR.
\end{proof}

\begin{definition}[FJR violation witness]\label{def:fjr-violation}
    A triple $(S,T,\ell)$ is a Hare-FJR violation witness for $W$ if
    \[
        |S|\ge\lambda_H|T|
        \qquad\text{and}\qquad
        u_i(T)\ge\ell>u_i(W)
        \quad\text{for every }i\in S.
    \]
    It is a Droop-FJR violation witness if the first inequality is replaced by
    \[
        |S|>\lambda_D|T|.
    \]
    A committee satisfies the respective FJR axiom if and only if it admits no corresponding violation witness.
\end{definition}

Every Hare- or Droop-FJR demand satisfies
\[
    1\le\ell\le|T|\le k.
\]
The utility condition gives $\ell\le|T|$. In the Hare case, $|S|\le n$ and $|S|\ge(n/k)|T|$ imply $|T|\le k$. In the Droop case, $|T|\ge k+1$ would imply $|S|>n$, which is impossible.

Finally, fix a voter order
\[
    i_1,i_2,\ldots,i_n
\]
and a candidate order
\[
    c_1,c_2,\ldots,c_m.
\]
The candidate order selects deterministically among affordable candidates, while the voter order specifies a deterministic capped-charging vector. These orders may affect payments and the resulting committee, but they do not affect either FJR guarantee.

\section{Construction of the Algorithm}\label{sec:algorithm}

The construction is parameterized by a common candidate price
\[
    \lambda\in\{\lambda_H,\lambda_D\}.
\]
Every voter receives one unit of initial budget. Under the Hare instantiation, the aggregate price of a witness set $T$ is $\lambda_H|T|$ and the demand condition is $|S|\ge\lambda_H|T|$. Under the Droop instantiation, the price is $\lambda_D|T|$ and the demand condition is the strict inequality $|S|>\lambda_D|T|$. These are aggregate accounting relations; they do not assert that $T$ is directly fundable subject to the approval constraints.

\subsection{State after each paid purchase}

Every voter starts with one unit of budget. Every candidate purchased during the paid part of the construction has the common price $\lambda$, chosen as either $\lambda_H$ or $\lambda_D$.

Let $t$ denote the number of paid purchases completed so far. After exactly $t$ paid purchases, the persistent algorithmic state is
\[
    \bigl(P^{(t)},\mathbf r^{(t)}\bigr),
    \qquad
    \mathbf r^{(t)}:=(r_i^{(t)})_{i\in N},
\]
where:
\begin{itemize}
    \item $P^{(t)}\subseteq C$ is the set of candidates purchased in the first $t$ paid steps; and
    \item $r_i^{(t)}\in\mathbb Q_{\ge0}$ is voter $i$'s remaining balance after those $t$ purchases.
\end{itemize}

Initially,
\[
    t=0,
    \qquad
    P^{(0)}=\varnothing,
    \qquad
    r_i^{(0)}=1 \quad(i\in N).
\]
At every point, $|P^{(t)}|=t$. All arithmetic and comparisons are exact.

If the candidate bought from state $t$ is $c^{(t)}$, let $x_i^{(t)}$ denote voter $i$'s charge for that purchase, and write
\[
    x_{i,c^{(t)}}:=x_i^{(t)}.
\]
Thus, $x_{i,c}$ denotes the unique historical payment made by voter $i$ when the paid candidate $c$ was selected. A voter is charged only for a candidate she approves.

\subsection{Decision snapshots, active voters, gaps, and offers}

The representation level $\ell$ of a potential FJR witness is not known in advance. The descending phases therefore inspect every possible target level, beginning with the most demanding one.

The algorithm processes thresholds
\[
    h=k,k-1,\ldots,1.
\]
Whenever the body of the while loop is entered at threshold $h$ with $t<k$, call the pair $(h,t)$ a \emph{decision snapshot}. Active sets, gaps, offers, and candidate supports are defined only at decision snapshots.

At a decision snapshot $(h,t)$, define the active-voter set
\[
    \Active_h^{(t)}
    :=
    \left\{
    i\in N:
    |A_i|\ge h
    \text{ and }
    u_i(P^{(t)})<h
    \right\}.
\]
Thus, an active voter approves at least $h$ candidates in total but currently has fewer than $h$ approved candidates in the paid set $P^{(t)}$.

A remaining balance alone records purchasing power but not how many additional approved winners the voter still needs. I therefore normalize the balance by this representation deficit. For each $i\in\Active_h^{(t)}$, define her remaining gap and offer by
\[
    \delta_{i,h}^{(t)}:=h-u_i(P^{(t)}),
    \qquad
    \alpha_{i,h}^{(t)}:=\frac{r_i^{(t)}}{\delta_{i,h}^{(t)}}.
\]
No offer is defined or used for an inactive voter. Since $u_i(P^{(t)})<h$ for every $i\in\Active_h^{(t)}$ and both quantities are integers,
\[
    \delta_{i,h}^{(t)}=h-u_i(P^{(t)})\ge1.
\]
Hence, the denominator in $\alpha_{i,h}^{(t)}$ is never zero.

For each unselected candidate $c\in C\setminus P^{(t)}$, define the total support available for $c$ at snapshot $(h,t)$ by
\[
    \Bid_h^{(t)}(c)
    :=
    \sum_{\substack{i\in\Active_h^{(t)}\\ c\in A_i}}
    \alpha_{i,h}^{(t)}.
\]
The candidate is affordable at snapshot $(h,t)$ when
\[
    \Bid_h^{(t)}(c)\ge\lambda.
\]

All active sets, gaps, offers, and candidate supports used in one purchase decision are computed from the same state snapshot $(P^{(t)},\mathbf r^{(t)})$. They remain fixed while that candidate is charged. Only after charging is complete does the algorithm create the persistent state $(P^{(t+1)},\mathbf r^{(t+1)})$.

If a phase $h$ ends because no candidate is affordable, its last decision snapshot is called its \emph{normal terminal snapshot}. If a purchase makes $t=k$, the algorithm stops the paid construction immediately; no new decision snapshot is formed, and in particular no post-purchase offer $\alpha_{i,h}^{(k)}$ is defined.

\subsection{Exact capped charging}

The charging cap is chosen so that an active voter never pays more for one representation unit than her current per-unit offer. This will ensure that, as long as she remains deficient, her later offer can dominate every payment she has already made.

Suppose candidate $c$ is affordable at decision snapshot $(h,t)$. Its active approvers are considered in the fixed voter order. Let $x_i^{(t)}$ be voter $i$'s charge for this purchase. The charging routine imposes the cap
\[
    0\le x_i^{(t)}\le \alpha_{i,h}^{(t)}
\]
and stops as soon as the aggregate charge reaches $\lambda$. Voters who are inactive or do not approve $c$ pay zero.

The new remaining balance is
\[
    r_i^{(t+1)}:=r_i^{(t)}-x_i^{(t)}
    \qquad(i\in N).
\]

\begin{algorithm}[H]
    \caption{Exact capped charging at decision snapshot $(h,t)$}
    \label{alg:charge}
    \begin{algorithmic}[1]
        \Procedure{CappedCharge}{$c,\Active_h^{(t)},\alpha_h^{(t)},\mathbf r^{(t)},\lambda$}
        \State $R\gets\lambda$ \Comment{unpaid remainder of the candidate price}
        \State $x_i^{(t)}\gets0$ for every $i\in N$
        \For{$i$ in the fixed voter order $i_1,i_2,\ldots,i_n$}
        \If{$R=0$}
        \State \textbf{break}
        \EndIf
        \If{$i\in\Active_h^{(t)}$ \textbf{and} $c\in A_i$}
        \State $x_i^{(t)}\gets\min\{\alpha_{i,h}^{(t)},R\}$
        \State $R\gets R-x_i^{(t)}$
        \EndIf
        \EndFor
        \ForAll{$i\in N$}
        \State $r_i^{(t+1)}\gets r_i^{(t)}-x_i^{(t)}$
        \EndFor
        \State \Return $(\mathbf x^{(t)},\mathbf r^{(t+1)})$
        \EndProcedure
    \end{algorithmic}
\end{algorithm}

Whenever the routine is called, affordability guarantees
\[
    \sum_{\substack{i\in\Active_h^{(t)}\\c\in A_i}}
    \alpha_{i,h}^{(t)}\ge\lambda.
\]
Therefore, the routine reaches $R=0$, and the candidate receives total payment exactly $\lambda$.

\subsection{Complete algorithm}

\begin{breakablealgorithm}
    \caption{Parameterized Descending-Budget Algorithm for FJR}
    \label{alg:main}
    \begin{algorithmic}[1]
        \Require An approval-based multiwinner election $E=(N,C,(A_i)_{i\in N},k)$, a price $\lambda\in\{\lambda_H,\lambda_D\}$, and fixed voter and candidate orders
        \Ensure A committee $W\subseteq C$ with $|W|=k$

        \State $t\gets0$
        \State $P^{(0)}\gets\varnothing$
        \State $r_i^{(0)}\gets1$ for every $i\in N$

        \For{$h=k,k-1,\ldots,1$}
        \While{$t<k$}
        \State $\Active_h^{(t)}\gets\{i\in N:|A_i|\ge h\text{ and }u_i(P^{(t)})<h\}$

        \ForAll{$i\in\Active_h^{(t)}$}
        \State $\delta_{i,h}^{(t)}\gets h-u_i(P^{(t)})$
        \State $\alpha_{i,h}^{(t)}\gets r_i^{(t)}/\delta_{i,h}^{(t)}$
        \EndFor

        \ForAll{$c\in C\setminus P^{(t)}$}
        \State $\Bid_h^{(t)}(c)\gets\displaystyle
            \sum_{\substack{i\in\Active_h^{(t)}\\c\in A_i}}\alpha_{i,h}^{(t)}$
        \EndFor

        \State $B_h^{(t)}\gets\{c\in C\setminus P^{(t)}:\Bid_h^{(t)}(c)\ge\lambda\}$
        \If{$B_h^{(t)}=\varnothing$}
        \State \textbf{break} \Comment{$\dagger$}
        \EndIf

        \State $c^{(t)}\gets$ the first candidate of $B_h^{(t)}$ in the fixed candidate order
        \State $(\mathbf x^{(t)},\mathbf r^{(t+1)})\gets
            \Call{CappedCharge}{c^{(t)},\Active_h^{(t)},\alpha_h^{(t)},\mathbf r^{(t)},\lambda}$
        \State $P^{(t+1)}\gets P^{(t)}\cup\{c^{(t)}\}$
        \State $t\gets t+1$
        \EndWhile

        \If{$t=k$}
        \State \textbf{break} \Comment{no post-purchase decision snapshot is formed}
        \EndIf
        \EndFor

        \State $D\gets$ the first $k-t$ candidates of $C\setminus P^{(t)}$ in the fixed candidate order
        \Comment{$D=\varnothing$ when $t=k$}
        \State $W\gets P^{(t)}\cup D$ \Comment{padding candidates are added without payment}
        \State \Return $W$
    \end{algorithmic}
\end{breakablealgorithm}
\begin{definition}
    A phase $h$ ends \textit{normally} if and only if it exits the while loop through $\dagger$.
\end{definition}

Each purchase changes both balances and representation deficits, so offers must be recomputed from the new persistent state. A phase stops precisely when its current deficient voters can no longer fund any remaining candidate at price $\lambda$.

After every paid purchase, the counter $t$ is incremented. If $t<k$, all state-dependent quantities are recomputed at the next decision snapshot. A threshold phase ends normally when no remaining candidate is affordable at that threshold.

\paragraph{Running example.}
Consider an election with six voters, committee size $k=4$, and candidate set
\[
    C=\{a,b,c,d,e\}.
\]
For the Hare-price instantiation, the candidate price is
\[
    \lambda=\lambda_H=\frac{n}{k}=\frac{6}{4}=\frac{3}{2}.
\]
Let the approval sets be
\[
    \begin{aligned}
        A_1=A_2 & =\{a,b\}, \\
        A_3=A_4 & =\{a,c\}, \\
        A_5     & =\{b,c\}, \\
        A_6     & =\{d\}.
    \end{aligned}
\]
Fix the voter order $1,2,\ldots,6$ and the candidate order
$a,b,c,d,e$.

This profile contains a genuinely weakly cohesive group. Let
\[
    S=\{1,2,3,4,5\},
    \qquad
    T=\{a,b,c\},
    \qquad
    \ell=2.
\]
Every voter in $S$ approves exactly two candidates in $T$, and
\[
    |S|=5
    \geq
    \lambda |T|
    =
    \frac{3}{2}\cdot 3
    =
    \frac{9}{2}.
\]
Hence, $S$ is Hare weakly $(2,T)$-cohesive. Notice, however, that
\[
    \bigcap_{i\in S} A_i=\varnothing,
\]
so $S$ is not ordinarily $2$-cohesive. This illustrates the distributed
agreement permitted by FJR.

Initially,
\[
    P^{(0)}=\varnothing,
    \qquad
    r^{(0)}=(1,1,1,1,1,1).
\]
No voter approves at least three candidates, so phases $h=4$ and $h=3$
end normally without a purchase.

At the first decision snapshot of phase $h=2$, voters $1,\ldots,5$ are
active. Each has gap $2$ and therefore offers
\[
    \alpha_{i,2}^{(0)}=\frac{1}{2}
    \qquad (i=1,\ldots,5).
\]
The candidate bids are
\[
    \begin{aligned}
        \operatorname{bid}_{2}^{(0)}(a)
         & =\frac{1}{2}+\frac{1}{2}
        +\frac{1}{2}+\frac{1}{2}
        =2,                                     \\
        \operatorname{bid}_{2}^{(0)}(b)
         & =\frac{1}{2}+\frac{1}{2}+\frac{1}{2}
        =\frac{3}{2},                           \\
        \operatorname{bid}_{2}^{(0)}(c)
         & =\frac{1}{2}+\frac{1}{2}+\frac{1}{2}
        =\frac{3}{2}.
    \end{aligned}
\]
All three candidates are affordable, and the fixed candidate order selects
$a$. The capped-charging routine charges voters $1$, $2$, and $3$ an amount
$1/2$ each, while voter $4$ pays zero. Thus,
\[
    P^{(1)}=\{a\},
    \qquad
    r^{(1)}
    =
    \left(
    \frac{1}{2},
    \frac{1}{2},
    \frac{1}{2},
    1,
    1,
    1
    \right).
\]

This purchase illustrates why offers must be recomputed. Voter $4$ approved
the selected candidate but paid nothing. Her balance remains $1$, while her
gap decreases from $2$ to $1$. Consequently, her offer increases from
\[
    \alpha_{4,2}^{(0)}=\frac{1}{2}
    \qquad\text{to}\qquad
    \alpha_{4,2}^{(1)}=1.
\]
At the new snapshot, the relevant offers are
\[
    \begin{aligned}
        \alpha_{1,2}^{(1)}
        =\alpha_{2,2}^{(1)}
        =\alpha_{3,2}^{(1)}
         & =\frac{1}{2}, \\
        \alpha_{4,2}^{(1)}
         & =1,           \\
        \alpha_{5,2}^{(1)}
         & =\frac{1}{2}.
    \end{aligned}
\]
Therefore,
\[
    \begin{aligned}
        \operatorname{bid}_{2}^{(1)}(b)
         & =\alpha_{1,2}^{(1)}
        +\alpha_{2,2}^{(1)}
        +\alpha_{5,2}^{(1)}
        =\frac{3}{2},          \\
        \operatorname{bid}_{2}^{(1)}(c)
         & =\alpha_{3,2}^{(1)}
        +\alpha_{4,2}^{(1)}
        +\alpha_{5,2}^{(1)}
        =2.
    \end{aligned}
\]
Both candidates are affordable, and the fixed order selects $b$. Voters
$1$, $2$, and $5$ each pay $1/2$, giving
\[
    P^{(2)}=\{a,b\},
    \qquad
    r^{(2)}
    =
    \left(
    0,
    0,
    \frac{1}{2},
    1,
    \frac{1}{2},
    1
    \right).
\]

Voters $1$ and $2$ now have utility $2$ and are inactive in phase $h=2$.
Voters $3$, $4$, and $5$ remain active, each with gap $1$. Their offers for
$c$ are respectively
\[
    \frac{1}{2},
    \qquad
    1,
    \qquad
    \frac{1}{2},
\]
so
\[
    \operatorname{bid}_{2}^{(2)}(c)=2.
\]
The algorithm purchases $c$. The charging routine assigns payment $1/2$ to
voter $3$ and payment $1$ to voter $4$; voter $5$ pays zero. Hence,
\[
    P^{(3)}=\{a,b,c\},
    \qquad
    r^{(3)}
    =
    \left(
    0,
    0,
    0,
    0,
    \frac{1}{2},
    1
    \right).
\]
At this point, every voter in $S$ has utility exactly $2$:
\[
    u_i(P^{(3)})=2
    \qquad (i\in S).
\]
Thus, no voter is active at the next snapshot of phase $h=2$, and the phase
ends normally.

In phase $h=1$, voters $1,\ldots,5$ are already represented and are inactive.
Only voter $6$ is active. She has balance $1$, gap $1$, and therefore offers
$1$ to candidate $d$. Since
\[
    \operatorname{bid}_{1}^{(3)}(d)
    =
    1
    <
    \lambda
    =
    \frac{3}{2},
\]
no candidate is affordable, and phase $h=1$ also ends normally.

The paid construction has selected only three candidates, so one seat remains.
The deterministic padding step adds the first unselected candidate in the
fixed candidate order, namely $d$. The final committee is therefore
\[
    W=\{a,b,c,d\}.
\]

The example illustrates three features of the construction. First, the set
$T=\{a,b,c\}$ witnesses weak cohesiveness even though the voters have no
commonly approved candidate. Second, a voter's offer may increase after an
approved candidate is purchased, because her representation gap can decrease
without a corresponding decrease in her balance. Third, the paid construction
may terminate with fewer than $k$ candidates, after which deterministic
padding completes the committee without affecting the preceding budget
process.

\subsection{Operational interpretation}

At a fixed threshold $h$ and decision snapshot $(h,t)$, every active voter divides her current remaining balance $r_i^{(t)}$ uniformly across the number $\delta_{i,h}^{(t)}$ of additional approved winners she still needs to reach utility $h$. The resulting amount $\alpha_{i,h}^{(t)}$ is her offer for each currently unselected candidate she approves. A candidate is purchased when the sum of its active approvers' offers reaches the selected common price $\lambda\in\{\lambda_H,\lambda_D\}$.

Higher thresholds are processed first. If no candidate is affordable at threshold $h$, the algorithm lowers the threshold to $h-1$. After the paid construction stops, any unfilled seats are padded deterministically and without charging voters. The padding set $D$ remains separate from the paid state $P^{(t)}$.

\section{Complexity}\label{sec:complexity}

I assume that the approval profile is represented by an $n \times m$ Boolean incidence matrix, ensuring that membership queries $c \in A_i$ execute in $O(1)$ time. The approval-set sizes $|A_i|$ are precomputed. For every voter $i$, I maintain the current paid utility:
\[
    s_i^{(t)} := u_i(P^{(t)}),
\]
along with a Boolean indicator for each selected candidate. Upon purchasing a candidate, all utility counters are updated in $O(n)$ time.

Every decision snapshot produces one of two outcomes: either a candidate is purchased, incrementing $t$ by one, or the current phase terminates normally. There are at most $k$ purchases and at most one normally terminating snapshot for each of the $k$ phases. Consequently, the algorithm evaluates at most $2k$ decision snapshots in total.

At any given decision snapshot, the active voters, their respective gaps, and their corresponding offers can be computed in $O(n)$ elementary operations. All candidate bids are subsequently calculated in $O(mn)$ rational additions by scanning the approval matrix. Identifying the first affordable candidate according to the fixed candidate ordering requires $O(m)$ comparisons, while a single capped-charging invocation—including the subsequent utility update—involves $O(n)$ additional operations. Thus, a single decision snapshot requires $O(mn)$ elementary rational operations, rendering the overall time complexity of the paid construction $O(kmn)$ such operations. The initial preprocessing and final padding steps require $O(mn)$ and $O(m)$ time, respectively, both of which are dominated by this bound.

It remains to establish an upper bound on the encoding length of the rational quantities. Let
\[
    L := \operatorname*{lcm}(1, 2, \ldots, k+1).
\]
I note that $L$ is introduced solely for analysis; the algorithm does not explicitly compute it. After $t$ paid purchases, the reduced denominator of every remaining balance $r_i^{(t)}$ divides $L^t$. Indeed, this claim holds trivially at $t = 0$. At a decision snapshot $(h, t)$, every positive gap is an integer in $\{1, \ldots, k\}$, and thus divides $L$. Consequently, the reduced denominator of every offer
\[
    \alpha_{i,h}^{(t)} = \frac{r_i^{(t)}}{\delta_{i,h}^{(t)}}
\]
divides $L^{t+1}$. The same property applies to every candidate bid, as it is a summation of offers whose denominators each divide $L^{t+1}$. The reduced denominator of $\lambda_H=n/k$ divides $k$, while that of $\lambda_D=n/(k+1)$ divides $k+1$; hence the denominator of either permitted price divides $L$. By induction through the capped-charging routine, every unpaid remainder and every individual charge has a reduced denominator dividing $L^{t+1}$. Subtracting such a charge from an existing balance yields a new balance whose reduced denominator similarly divides $L^{t+1}$.

Because $t \le k$ and
\[
    \log L \le \log((k+1)!) = O(k \log k),
\]
every denominator possesses a bit length bounded by $O(k^2 \log k)$. Furthermore, since
\[
    0 \le r_i^{(t)}, \alpha_{i,h}^{(t)}, x_i^{(t)} \le 1,
\]
while every candidate bid, unpaid remainder, and candidate price is bounded above by $n$, each numerator and denominator has a maximum bit length of
\[
    B = O(\log n + k^2 \log k).
\]
By representing all rational numbers in reduced form, every exact addition, subtraction, comparison, division by an integer gap, minimum operation, and normalization step can be executed in $\mathrm{poly}(B)$ bit operations. The overall bit complexity is therefore
\[
    O\left(kmn \cdot \mathrm{poly}(B)\right),
\]
which is polynomial in the representation size of the explicit election instance.

The snapshot bound above guarantees algorithmic termination. Upon termination, the paid set contains size $t \le k$. Given $m \ge k$, I have
\[
    |C \setminus P^{(t)}| = m - t \ge k - t,
\]
ensuring that the deterministic padding step appends exactly $k - t$ distinct candidates. The resulting committee therefore has a size of precisely $k$.
\section{Proofs of the FJR guarantees}\label{sec:proof}

\paragraph{Proof map.}
The common part of the proof is a row--column accounting argument valid for either permitted candidate price. I first establish budget conservation and show that an active voter's current offer dominates all of her historical payments. These facts allow me to assign a mixed cost to any target bundle: selected candidates are priced by historical payments and unselected candidates by current offers, yielding a unit lower bound for every deficient voter row. Normal termination gives the complementary candidate-column bounds. The Hare and Droop proofs differ only when the paid set reaches size $k$ before the witness phase has a normal terminal snapshot: the Hare price exhausts all budgets, whereas the Droop price leaves total residual budget $n/(k+1)$ and requires a witness-specific lower bound.

\subsection{Common accounting lemmas}

The first lemma fixes the budget ledger used on both sides of the double counting.

\begin{lemma}\label{lem:accounting}
    For every persistent state $t$ and every voter $i$,
    \[
        r_i^{(t)}+\sum_{c\in A_i\cap P^{(t)}}x_{i,c}=1.
    \]
    Moreover,
    \[
        \sum_{i\in N}r_i^{(t)}=n-t\lambda.
    \]
    In particular, all balances are nonnegative.
\end{lemma}

\begin{proof}
    At a decision snapshot, an active voter has gap at least one, so
    \[
        \alpha_{i,h}^{(t)}=\frac{r_i^{(t)}}{\delta_{i,h}^{(t)}}\le r_i^{(t)}.
    \]
    The charging cap therefore implies $0\le x_i^{(t)}\le r_i^{(t)}$, and balances remain nonnegative. Since voters are charged only for approved selected candidates, iterating the update
    \[
        r_i^{(t+1)}=r_i^{(t)}-x_i^{(t)}
    \]
    gives the first identity. Every paid candidate receives total payment exactly $\lambda$, so summing over voters gives
    \[
        \sum_{i\in N}r_i^{(t)}=n-t\lambda.
    \]
\end{proof}

The next invariant is what makes historical payments comparable with current purchasing power. It will later allow payments made outside a target bundle to be replaced, in the accounting argument, by current offers on unselected target candidates. The lemma is deliberately stated only at actual decision snapshots and does not introduce a hypothetical offer after a purchase that fills the paid committee.

\begin{lemma}\label{lem:historical_pay}
    Let $(h,t)$ be a decision snapshot reached by the algorithm. If voter $i$ is active at that snapshot, then every historical payment made by $i$ for a selected approved candidate satisfies
    \[
        x_{i,c}\le \alpha_{i,h}^{(t)}
        \qquad\text{for every }c\in A_i\cap P^{(t)}.
    \]
\end{lemma}

\begin{proof}
    I induct over the chronological sequence of decision snapshots reached by the algorithm. At the first snapshot, $P^{(0)}=\varnothing$, so the statement is vacuous.

    First consider two consecutive decision snapshots in the same phase $h$. The first is $(h,t)$, a candidate $c^{(t)}$ is purchased there, and the next is $(h,t+1)$. Such a next snapshot exists only when $t+1<k$. Fix a voter $i$ who is active at $(h,t+1)$. Since utility never decreases, she was also active at $(h,t)$.

    If $c^{(t)}\notin A_i$, then voter $i$ is not charged, and neither her balance nor her utility changes. Hence,
    \[
        \alpha_{i,h}^{(t+1)}=\alpha_{i,h}^{(t)},
    \]
    while $A_i\cap P^{(t+1)}=A_i\cap P^{(t)}$. The claim follows from the induction hypothesis.

    Now suppose $c^{(t)}\in A_i$. Write
    \[
        d:=h-u_i(P^{(t)}),
        \qquad
        r:=r_i^{(t)},
        \qquad
        a:=\alpha_{i,h}^{(t)}=\frac{r}{d},
    \]
    and let $x:=x_i^{(t)}$. Since $i$ remains active after the purchase, $d\ge2$. The charging rule gives $0\le x\le a$, and at the next decision snapshot her offer is
    \[
        a':=\alpha_{i,h}^{(t+1)}=\frac{r-x}{d-1}.
    \]
    Using $r=ad$,
    \[
        a'-a
        =\frac{r-x}{d-1}-\frac{r}{d}
        =\frac{a-x}{d-1}
        \ge0.
    \]
    Thus, every older payment, which is at most $a$ by the induction hypothesis, is at most $a'$. The new payment also satisfies $x\le a\le a'$. Hence, the bound is preserved.

    It remains to consider a transition from the normal terminal snapshot of phase $h+1$ to the first decision snapshot of phase $h$. No candidate, balance, or historical payment changes at this transition. Fix a voter $i$ who is active in phase $h$.

    If $|A_i|\ge h+1$, then $u_i(P^{(t)})<h<h+1$, so voter $i$ was also active at the normal terminal snapshot of phase $h+1$. By the induction hypothesis, all of her historical payments are at most $\alpha_{i,h+1}^{(t)}$. Since the balance and utility are unchanged and the denominator decreases,
    \[
        \alpha_{i,h}^{(t)}
        =\frac{r_i^{(t)}}{h-u_i(P^{(t)})}
        \ge
        \frac{r_i^{(t)}}{h+1-u_i(P^{(t)})}
        =\alpha_{i,h+1}^{(t)}.
    \]
    Thus, the required bounds remain valid.

    If $|A_i|=h$, then voter $i$ was ineligible in every previously processed, higher-threshold phase. Since only active voters are charged, she has made no historical payment, and the statement holds trivially.

    These cases cover every transition between consecutive decision snapshots. No induction step is required after a purchase that makes $t=k$, because the algorithm forms no further decision snapshot.
\end{proof}

\begin{corollary}\label{cor:positive_balance}
    If voter $i$ is active at a decision snapshot $(h,t)$, then
    \[
        r_i^{(t)}>0.
    \]
\end{corollary}

\begin{proof}
    Suppose instead that $r_i^{(t)}=0$. Since $i$ is active, her gap is positive and therefore $\alpha_{i,h}^{(t)}=0$. Lemma~\ref{lem:historical_pay} and nonnegativity of payments imply
    \[
        x_{i,c}=0
        \qquad\text{for every }c\in A_i\cap P^{(t)}.
    \]
    The budget identity in Lemma~\ref{lem:accounting} would then give $1=0$, a contradiction.
\end{proof}

To bridge the voter-side and candidate-side estimates, I introduce a mixed cost metric. Historical payments govern selected candidates, whereas terminal offers govern unselected candidates; the proposed cost function applies the corresponding quantity to each case.

\begin{definition}[Mixed Cost]
    At a decision snapshot $(h,t)$, for an active voter $i$ and a set of candidates $B \subseteq A_i$, its cost is defined as:
    \[
        \Cost_{i,h}^{(t)}(B)
        :=
        \left( \sum_{c \in B \cap P^{(t)}} x_{i,c} \right)
        +
        |B \setminus P^{(t)}| \, \alpha_{i,h}^{(t)}.
    \]
\end{definition}

Actual historical payments price selected candidates, whereas the offer at the current decision snapshot prices unselected candidates. Equivalently, interpreting voters as rows and candidates as columns, this metric measures the load of voter $i$ across the columns in $B$.

The core objective of this definition is to guarantee that every active voter can project her full unit budget onto any approved bundle large enough to satisfy the active threshold. Consequently, payments committed to the bundle are preserved, while current offers compensate for both the remaining budget and prior expenditures incurred outside the bundle.
\begin{lemma}\label{lem:cost_lowerbound}
    At every decision snapshot $(h,t)$,
    \[
        \Cost_{i,h}^{(t)}(B)\ge1
    \]
    for every active voter $i$ and every $B\subseteq A_i$ with $|B|\ge h$.
\end{lemma}

\begin{proof}
    Put
    \[
        u:=u_i(P^{(t)}),
        \qquad
        r:=r_i^{(t)},
        \qquad
        a:=\alpha_{i,h}^{(t)}=\frac{r}{h-u},
    \]
    and let
    \[
        Q:=B\cap P^{(t)},
        \qquad q:=|Q|.
    \]
    By Lemma~\ref{lem:accounting}, the total payment on the $u$ selected candidates approved by voter $i$ is $1-r$. Exactly $u-q$ of those candidates lie outside $Q$. Lemma~\ref{lem:historical_pay} bounds each corresponding payment by $a$, so
    \[
        \sum_{c\in Q}x_{i,c}
        \ge 1-r-(u-q)a.
    \]
    Moreover,
    \[
        |B\setminus P^{(t)}|=|B|-q\ge h-q.
    \]
    Consequently,
    \begin{align*}
        \Cost_{i,h}^{(t)}(B)
         & \ge 1-r-(u-q)a+(h-q)a   \\
         & =1-r+(h-u)a             \\
         & =1-r+(h-u)\frac{r}{h-u} \\
         & =1.
    \end{align*}
\end{proof}

The preceding lemma supplies the row lower bound. The next observation supplies the corresponding column capacities: a selected column collected exactly its price, whereas an unselected column at a normal terminal snapshot falls strictly short of that price.

\begin{lemma}\label{lem:terminal_status}
    Suppose phase $h$ ends normally at decision snapshot $(h,\tau_h)$. Then
    \[
        \Bid_h^{(\tau_h)}(c)<\lambda
        \qquad\text{for every }c\in C\setminus P^{(\tau_h)}.
    \]
    Every candidate in $P^{(\tau_h)}$ received total historical payment exactly $\lambda$ when it was purchased.
\end{lemma}

\begin{proof}
    At a normal terminal snapshot, the affordable set is empty. Since affordability is defined by the weak inequality $\Bid_h^{(\tau_h)}(c)\ge\lambda$, every unselected candidate must satisfy the displayed strict inequality. The second statement follows directly from the exact capped-charging routine.
\end{proof}

\subsection{Hare-FJR guarantee}

The main double-counting argument is carried out at a normal terminal snapshot of the witness threshold. Under the Hare price, the only obstruction is that the paid construction may fill all $k$ seats before such a snapshot is reached. In that event, the total budget is exhausted and every voter must meet the target of the final purchase phase.

\begin{lemma}\label{lem:hare_early_quit}
    Run Algorithm~\ref{alg:main} with $\lambda=\lambda_H$, and suppose the paid set first reaches size $k$ through a purchase made in phase $g$. Then, in the resulting persistent state,
    \[
        u_i(P^{(k)})\ge g
        \qquad\text{for every }i\in N.
    \]
\end{lemma}

\begin{proof}
    In the Hare instantiation, $\lambda=\lambda_H=n/k$. Lemma~\ref{lem:accounting} therefore gives total remaining balance $n-k\lambda_H=0$. Since balances are nonnegative, every balance is zero in the resulting state:
    \[
        r_i^{(k)}=0
        \qquad\text{for every }i\in N.
    \]
    Suppose, for a contradiction, that some voter $i$ satisfies
    \[
        u_i(P^{(k)})<g.
    \]

    If $|A_i|<g$, then voter $i$ was ineligible in phase $g$ and in every previously processed higher-threshold phase. She was therefore never charged, so $r_i^{(k)}=1$, contradicting $r_i^{(k)}=0$.

    Now suppose $|A_i|\ge g$. The final purchase was made from the decision snapshot $(g,k-1)$. Since
    \[
        u_i(P^{(k-1)})\le u_i(P^{(k)})<g,
    \]
    voter $i$ was active at that snapshot. Corollary~\ref{cor:positive_balance} gives
    \[
        r_i^{(k-1)}>0.
    \]
    Let $c^{(k-1)}$ be the final paid candidate. If $c^{(k-1)}\notin A_i$, then $i$ pays zero in the final purchase, so
    \[
        r_i^{(k)}=r_i^{(k-1)}>0,
    \]
    a contradiction.

    If $c^{(k-1)}\in A_i$, then the assumption $u_i(P^{(k)})<g$ implies that the pre-purchase gap
    \[
        d:=g-u_i(P^{(k-1)})
    \]
    satisfies $d\ge2$. By the charging cap,
    \[
        x_i^{(k-1)}
        \le \alpha_{i,g}^{(k-1)}
        =\frac{r_i^{(k-1)}}{d}
        <r_i^{(k-1)}.
    \]
    Hence,
    \[
        r_i^{(k)}
        =r_i^{(k-1)}-x_i^{(k-1)}
        >0,
    \]
    again contradicting $r_i^{(k)}=0$. Therefore, no such voter exists.
\end{proof}

\begin{theorem}[Hare-FJR guarantee]\label{thm:hare}
    For every explicitly represented approval-based multiwinner election $E$, Algorithm~\ref{alg:main} with price $\lambda=\lambda_H$ terminates and returns a size-$k$ committee satisfying Hare-FJR.
\end{theorem}
\begin{proof}
    Because $P^{(t)}\subseteq W$, padding can only increase voters' approval utilities.
    The proof may therefore analyze the paid set at earlier decision snapshots while using an assumed violation of the final committee $W$ to infer that the relevant voters were still deficient at those snapshots.

    Suppose, for a contradiction, that the output committee $W$ violates Hare-FJR. Then there exist a nonempty set $S\subseteq N$, a set $T\subseteq C$, and a positive integer $\ell$ such that
    \[
        k|S|\ge n|T|
        \qquad\text{and}\qquad
        u_i(T)\ge\ell>u_i(W)
        \quad\text{for every }i\in S.
    \]
    Writing $\lambda_H=n/k$, the quota condition is
    \[
        |S|\ge\lambda_H|T|.
    \]
    Since $|S|\le n$, it follows that $|T|\le k$. Also, $u_i(T)\ge\ell$ implies $\ell\le|T|$, so
    \[
        1\le\ell\le|T|\le k.
    \]

    I first show that phase $\ell$ ends normally. If the paid construction never reaches size $k$, then phase $\ell$ cannot terminate by filling all $k$ paid seats and therefore must end normally.
    Otherwise, let $g$ be the phase in which the paid set first reaches size $k$.
    By Lemma~\ref{lem:hare_early_quit}, $u_i(W)\ge u_i(P^{(k)})\ge g$ for every voter $i$. On the other hand, every $i\in S$ satisfies $u_i(W)<\ell$, so $g<\ell$. Since the phases are processed in descending order, phase $\ell$ is completed before phase $g$ begins.
    Because the paid set first reaches size $k$ only in phase $g$, phase $\ell$ cannot terminate by filling the paid committee. It must therefore end normally.

    Let $(\ell,\tau)$ be the normal terminal snapshot of phase $\ell$, and abbreviate
    \[
        P:=P^{(\tau)},
        \qquad
        \alpha_i:=\alpha_{i,\ell}^{(\tau)}
        \quad(i\in\Active_\ell^{(\tau)}).
    \]
    For every $i\in S$,
    \[
        |A_i|\ge u_i(T)\ge\ell
    \]
    and, because $P\subseteq W$,
    \[
        u_i(P)\le u_i(W)<\ell.
    \]
    Hence, every voter in $S$ is active at the terminal snapshot.

    I now form the row--column accounting matrix implicitly. For voter $i$, selected target candidates are charged at their historical payments and unselected target candidates at the current offer. For each $i\in S$, define
    \begin{align*}
         & B_i:=A_i\cap T,                               \\
         & X:=\sum_{i\in S}\sum_{c\in B_i\cap P}x_{i,c}, \\
         & Y:=\sum_{i\in S}|B_i\setminus P|\,\alpha_i.
    \end{align*}
    Then $|B_i|=u_i(T)\ge\ell$. Summing by voter rows and applying Lemma~\ref{lem:cost_lowerbound}, I obtain
    \begin{align*}
        X+Y & = \left(\sum_{i\in S}\sum_{c\in B_i\cap P}x_{i,c}\right)+\left(\sum_{i\in S}|B_i\setminus P|\,\alpha_i\right) \\
            & =\sum_{i\in S}\left(\bigl(\sum_{c\in B_i\cap P}x_{i,c}\bigr)+|B_i\setminus P|\,\alpha_i\right)                \\
            & =\sum_{i\in S}\Cost_{i,\ell}^{(\tau)}(B_i)\geq |S|.
    \end{align*}

    I next sum the same load by candidate columns. For the selected part, $X$, Lemma~\ref{lem:terminal_status} gives
    \begin{align*}
        X & =\sum_{i\in S}\sum_{c\in A_i\cap T\cap P}x_{i,c} \\
          & =\sum_{c\in T\cap P}\sum_{\substack{i\in S       \\c\in A_i}}x_{i,c}\quad(\text{rearrange terms})\\
          & \le\sum_{c\in T\cap P}\sum_{i\in N}x_{i,c}       \\
          & =\lambda_H|T\cap P|.
    \end{align*}

    For the unselected part, $Y$, rearranging terms gives
    \begin{align*}
        Y & =\sum_{i\in S}|A_i\cap T\setminus P|\,\alpha_i    \\
          & =\sum_{c\in T\setminus P}\sum_{\substack{i\in S   \\c\in A_i}}\alpha_i\quad(\text{rearrange terms})\\
          & \le\sum_{c\in T\setminus P}\Bid_\ell^{(\tau)}(c).
    \end{align*}
    The set $T\setminus P$ is nonempty: if $T\subseteq P$, then every $i\in S$ would satisfy
    \[
        u_i(W)\ge u_i(P)\ge u_i(T)\ge\ell,
    \]
    contrary to the assumed violation. Therefore, Lemma~\ref{lem:terminal_status} implies
    \[
        Y<\lambda_H|T\setminus P|.
    \]
    Combining the row bound, the selected-column bound, the strict unselected-column bound, and the quota condition yields
    \[
        |S|
        \le X+Y
        <\lambda_H\bigl(|T\cap P|+|T\setminus P|\bigr)
        =\lambda_H|T|
        \le|S|,
    \]
    a contradiction. Thus, $W$ satisfies Hare-FJR.
\end{proof}

\subsection{Droop-FJR guarantee}

For the Droop instantiation, the normal-terminal-snapshot argument remains unchanged except for the strict demand inequality. The full-seat case is different: after $k$ purchases the total residual budget is $n/(k+1)$ rather than zero. The following lemma replaces Lemma~\ref{lem:hare_early_quit}.

\begin{lemma}[Droop full-seat lemma]\label{lem:droop_early_quit}
    Run Algorithm~\ref{alg:main} with $\lambda=\lambda_D$, and suppose that the paid set first reaches size $k$ through a purchase made in phase $g$. If $(S,T,\ell)$ is a Droop-FJR violation witness for the resulting committee, then
    \[
        \ell>g.
    \]
\end{lemma}

\begin{proof}
    Suppose, for a contradiction, that $\ell\le g$. Since the paid set has size $k$, the final committee is $W=P^{(k)}$, and every $i\in S$ satisfies
    \[
        u_i(P^{(k)})=u_i(W)<\ell\le g.
    \]
    I first show that each voter $i\in S$ retains at least
    \[
        \frac{g-\ell+1}{g}
    \]
    units of budget after the final purchase.

    Fix $i\in S$. If $|A_i|<g$, then voter $i$ was inactive in phase $g$ and in every previously processed higher-threshold phase. She was never charged before the algorithm stopped, so
    \[
        r_i^{(k)}=1\ge\frac{g-\ell+1}{g}.
    \]

    Now suppose $|A_i|\ge g$. The final purchase was made from the decision snapshot $(g,k-1)$. Put
    \[
        u:=u_i(P^{(k-1)}),
        \qquad
        d:=g-u,
        \qquad
        r:=r_i^{(k-1)},
        \qquad
        a:=\alpha_{i,g}^{(k-1)}=\frac{r}{d}.
    \]
    Because $u_i(P^{(k)})<\ell\le g$, voter $i$ was active at this snapshot. By Lemma~\ref{lem:historical_pay}, each of her $u$ historical payments is at most $a$. The individual budget identity in Lemma~\ref{lem:accounting} therefore gives
    \[
        1
        =r+\sum_{c\in A_i\cap P^{(k-1)}}x_{i,c}
        \le da+ua
        =ga,
    \]
    and hence
    \[
        a\ge\frac1g.
    \]

    Let $c^{(k-1)}$ be the final paid candidate. If $c^{(k-1)}\notin A_i$, then $u_i(P^{(k)})=u<\ell$, so $d=g-u\ge g-\ell+1$, and
    \[
        r_i^{(k)}=r=da
        \ge\frac{g-\ell+1}{g}.
    \]
    If $c^{(k-1)}\in A_i$, then $u_i(P^{(k)})=u+1<\ell$, so $d\ge g-\ell+2$. The charging cap gives $x_i^{(k-1)}\le a$, and therefore
    \[
        r_i^{(k)}
        =da-x_i^{(k-1)}
        \ge(d-1)a
        \ge\frac{g-\ell+1}{g}.
    \]
    Thus, in all cases,
    \[
        \sum_{i\in N}r_i^{(k)}
        \ge\sum_{i\in S}r_i^{(k)}
        \ge |S|\frac{g-\ell+1}{g}.
    \]

    Since $(S,T,\ell)$ is a Droop-FJR violation witness and $|T|\ge\ell$,
    \[
        |S|>\lambda_D|T|\ge\lambda_D\ell.
    \]
    Moreover, for $1\le\ell\le g$,
    \[
        \ell(g-\ell+1)-g
        =(\ell-1)(g-\ell)
        \ge0,
    \]
    so
    \[
        \frac{\ell(g-\ell+1)}{g}\ge1.
    \]
    Combining these estimates yields
    \[
        \sum_{i\in N}r_i^{(k)}
        >\lambda_D\frac{\ell(g-\ell+1)}{g}
        \ge\lambda_D.
    \]
    On the other hand, Lemma~\ref{lem:accounting} and $\lambda_D=n/(k+1)$ give the exact total
    \[
        \sum_{i\in N}r_i^{(k)}
        =n-k\lambda_D
        =\frac{n}{k+1}
        =\lambda_D,
    \]
    a contradiction. Hence $\ell>g$.
\end{proof}

\begin{theorem}[Droop-FJR guarantee]\label{thm:droop}
    For every explicitly represented approval-based multiwinner election $E$, Algorithm~\ref{alg:main} with price $\lambda=\lambda_D$ terminates and returns a size-$k$ committee satisfying Droop-FJR.
\end{theorem}

\begin{proof}
    Termination and the size-$k$ output guarantee follow from Section~\ref{sec:complexity}. Suppose, for a contradiction, that the output committee $W$ admits a Droop-FJR violation witness $(S,T,\ell)$. Thus,
    \[
        |S|>\lambda_D|T|
        \qquad\text{and}\qquad
        u_i(T)\ge\ell>u_i(W)
        \quad\text{for every }i\in S.
    \]
    As observed in Section~\ref{sec:preliminaries}, $1\le\ell\le|T|\le k$.

    I first show that phase $\ell$ ends normally. If the paid construction never reaches size $k$, this is immediate. Otherwise, let $g$ be the phase in which the paid set first reaches size $k$. Lemma~\ref{lem:droop_early_quit} rules out $\ell\le g$, so $\ell>g$. Since phases are processed in descending order, phase $\ell$ was completed before phase $g$ began. It could not have ended by filling all $k$ paid seats, and therefore ended normally.

    Let $(\ell,\tau)$ be the normal terminal snapshot of phase $\ell$, and abbreviate
    \[
        P:=P^{(\tau)},
        \qquad
        \alpha_i:=\alpha_{i,\ell}^{(\tau)}
        \quad(i\in\Active_\ell^{(\tau)}).
    \]
    Because $P\subseteq W$, every $i\in S$ satisfies
    \[
        |A_i|\ge u_i(T)\ge\ell
        \qquad\text{and}\qquad
        u_i(P)\le u_i(W)<\ell.
    \]
    Hence every voter in $S$ is active at the terminal snapshot.

    For each $i\in S$, put $B_i:=A_i\cap T$ and define
    \[
        X:=\sum_{i\in S}\sum_{c\in B_i\cap P}x_{i,c},
        \qquad
        Y:=\sum_{i\in S}|B_i\setminus P|\,\alpha_i.
    \]
    Since $|B_i|=u_i(T)\ge\ell$, Lemma~\ref{lem:cost_lowerbound} gives
    \[
        X+Y\ge|S|.
    \]
    Summing by candidate columns and applying Lemma~\ref{lem:terminal_status} gives
    \[
        X\le\lambda_D|T\cap P|.
    \]
    Moreover, $T\setminus P\neq\varnothing$, because $T\subseteq P$ would imply $u_i(W)\ge u_i(T)\ge\ell$ for every $i\in S$. Therefore,
    \[
        Y<\lambda_D|T\setminus P|.
    \]
    Combining the row and column estimates with the strict Droop demand condition yields
    \[
        |S|
        \le X+Y
        <\lambda_D|T|
        <|S|,
    \]
    a contradiction. Thus, $W$ satisfies Droop-FJR.
\end{proof}

\section{Conclusion \& Future Work}

I have presented a single descending-budget construction with two candidate-price instantiations. With price $\lambda_H=n/k$, it computes a Hare-FJR committee; with price $\lambda_D=n/(k+1)$, it computes a committee satisfying the Droop-FJR axiom of Casey and Elkind. Both variants run in polynomial time and use the same active-voter, gap-normalized-offer, capped-charging, and padding rules. The Droop extension changes only the candidate price in the algorithm, but its proof requires a different treatment of the case in which all $k$ paid seats are filled: instead of complete budget exhaustion, the argument uses the exact residual total $n/(k+1)$ and a lower bound on the balances retained by the voters in a putative violation witness.

Several open directions remain. Although both FJR guarantees hold for every fixed voter and candidate order, these orders can affect the payment distribution and the resulting committee; the deterministic implementation is therefore not inherently anonymous or neutral. It would be useful to identify more symmetric charging and candidate-selection procedures preserving the same invariants. Other natural questions include optimizing secondary welfare objectives among the resulting FJR committees and determining whether related budget-accounting ideas can yield efficient guarantees approaching core stability in approval-based multiwinner elections.

\section*{Usage of GenAI}

The core idea behind this proof grew out of my exploration, beginning in February, of several open problems related to EJR. 
During this process, I made extensive use of frontier models such as Gemini 3.1 Pro and GPT-5.4/5.5 Pro. 
Although these models did not directly produce any useful line of attack, I noticed through these conversations that the Greedy Justified Candidate Rule repeatedly came up in LLM-generated attempts at several different problems. 
This observation led me, in June, to organize the relevant ideas more systematically and investigate whether GJCR could be combined with the Method of Equal Shares to obtain a new algorithm. 
In July, I arrived at an early version of the algorithm presented here, which consistently appeared to produce committees satisfying FJR. 
I then began working toward a proof of the algorithm's correctness. 
During this stage, GPT-5.6 Sol helped me identify and fix many errors in the proof and improve various aspects of the exposition.

\bibliographystyle{alphaurl}
\bibliography{reference}

\end{document}